\documentclass[11pt,a4paper,leqno]{amsart}

\newcommand\dela[1]{}

\usepackage{a4wide}

\usepackage{amssymb}
\usepackage{verbatim}
\usepackage{times}

\usepackage{mathrsfs}

\usepackage{verbatim}
\usepackage[colorlinks,linkcolor={blue},
citecolor={blue},urlcolor={red},]{hyperref}
\usepackage{hyperref}

\input{colordvi}

\theoremstyle{plain}

\newtheorem{theorem}{Theorem}[section]

\theoremstyle{remark}

\newtheorem{remark}[theorem]{Remark}

\newtheorem{question}[theorem]{Question}

\theoremstyle{plain}

\newtheorem{corollary}[theorem]{Corollary}

\newtheorem{lemma}[theorem]{Lemma}

\newtheorem{proposition}[theorem]{Proposition}

\newtheorem{definition}[theorem]{Definition}

\newenvironment{del}[1]{}{\par\vspace{6pt}}

\numberwithin{equation}{section}

\def\todown{\searrow}

\newcommand{\lb}{\langle}

\newcommand{\rb}{\rangle}

\newcommand\diver{\operatorname{div}}

\newcommand\curl{\operatorname{curl}}

\newcommand{\eps}{\varepsilon}

\newcommand{\al}{\alpha}

\def\Nat{{\mathbb N}}

\def\Rnu{{\mathbb R}}

\def\Rd{{\Rnu^d~}}

\def\liml{\lim\limits}

\def\liminfl{\liminf\limits}

\def\limsupl{\limsup\limits}

\def\suml{\sum\limits}

\def\supl{\sup\limits}

\def\intl{\int\limits}

\def\supl{\mathop{\sup}\limits}

\usepackage{colortbl}
\newcommand{\blue}{\color{blue}}

\begin{document}
\title[Stochastic Burgers equation]{Multidimensional stochastic Burgers equation}

\author[Z. Brze{\'z}niak, B. Goldys, M. Neklyudov]{ Zdzis{\l}aw Brze{\'z}niak, Ben Goldys and Misha Neklyudov}

\address{
School of Mathematics and Statistics,
University of NSW, Sydney,  Australia\\
Department of Mathematics, University of York, Heslington,
UK\\
Department of Mathematics, University of York, Heslington,
UK
}

\date{\today}

\begin{abstract}
We consider multidimensional stochastic
Burgers equation on the torus $\mathbb{T}^d$ and the whole space
$\Rd$. In both cases we show that for positive viscosity $\nu>0$ there exists a unique
strong global solution in $L^p$ for $p>d$. In the case of torus we also
establish a uniform in $\nu$ a priori estimate and consider a limit $\nu\todown 0$
for potential solutions. In the case of $\Rd$ uniform with respect to $\nu$ a priori estimate
established if a Beale-Kato-Majda type condition is satisfied.
\end{abstract}
\maketitle
\section{Introduction}
The aim of this paper is to study the existence and the uniqueness of solutions to the multidimensional stochastic Burgers equation
of the following form:
\begin{equation}\label{eq1}
\left\{\begin{array}{ll}
\frac{\partial u}{\partial t}=\nu\Delta u+u\cdot\nabla u+f+\xi,&t>0,\,\, x\in\mathcal O,\\
u(0,x)=u_0(x),&x\in\mathcal O,
\end{array}\right.
\end{equation}
where either $\mathcal O=\Rd$ or $\mathcal O=\mathbb T^d$. In the equation above $f$ is a deterministic force and $\xi$ is a multidimensional noise, white in time and correlated in space. We do not assume that $u_0$, $f$ and $\xi$ are of gradient form. The parameter $\nu>0$ is known as viscosity. In this paper we will also study the limit of solutions to \eqref{eq1} when $\nu\to 0$.

Equation \eqref{eq1} has been proposed by Burgers \cite{burgers} as a toy model for
turbulence, see also Weinan \cite{Weinan-2000}. Later, numerous applications were found in Astrophysics and Statistical Physics. For an interesting review of applications and problems related to equation \eqref{eq1}, see \cite{bec} and references therein. The Burgers equation with data of non-potential type arises in many areas of Physics, including gas dynamics and the theory of inelastic granular media, see for example \cite{naim}. The theory of equation \eqref{eq1} in the non-potential case is largely a terra incognita, see the review \cite{bec}, where a variety of open problems can be found. This paper and the preceding work \cite{Goldys-Neklyudov} by the second and third named authours are the first steps towards answering some of these questions.

One dimensional stochastic Burgers equation  has been fairly well studied. Da Prato, Debussche, Temam \cite{DaPratoTemam-1994}, see also Bertini, Cancrini and Jona-Lasinio \cite{BertiniJona-Lasinio-1994}, showed the existence of a unique global solution for one dimensional Burgers equation with additive noise. The existence and uniqueness results have been extended to the case of multiplicative noise by Da Prato, Gatarek \cite{DaPratoGatarek-1995} and Gy\"{o}ngy, Nualart \cite{GN1999}. 

Multidimensional Burgers equation has been studied much less comprehensively. Kiselev, Ladyzhenskaya \cite{[Kis+Lad_1957]} proved the existence and uniqueness of a global solution to the deterministic Burgers equation a bounded domain $\mathcal O$ in the class of functions $L^{\infty}(0,T;L^{\infty}(\mathcal{O}))\cap L^2(0,T; H^{1,2}_0(\mathcal{O}))$. The main idea of their proof is to apply maximum principle to deduce a priori estimates similar to the a priori estimates  for the Navier-Stokes equation. Ton \cite{BuiAnTon1975} established convergence of solutions on small time interval when we take the limit $\nu\to 0$ and when the initial condition is zero.
\par\noindent
The assumption that the initial condition and force have gradient form considerably simplifies analysis of the Burgers equation. It is well known that in this case one can apply the
Hopf-Cole transformation ((\cite{Hopf1950}, \cite{Cole1951})) to reduce the multidimensional Burgers equation either to the heat equation or to
the Hamilton-Jacobi equation, see for example \cite{DirrSouganidis-2005}.
The number of works on the Hopf-Cole transformation is huge and we will not try to list them all here. We only mention
Dermoune \cite{Dermoune-1999}, where the Hopf-Cole transformation is used to show the existence of solution to the stochastic multidimensional Burgers equation with additive noise.
Khanin et al \cite{GomesIturriaga-2000}  proved the existence of the so called quasi stationary solution by
the Hopf-Cole transformation and Stochastic Lax formula, thus partially extending to many dimensions an important paper \cite{sinai} by Sinai. This  approach however has certain intrinsic problems. In particular, it seems difficult to find an a priori estimate for the
solution without additional assumptions on the initial condition as in Dermoune \cite{Dermoune-1999} p. 303, Theorem 4.2.
Hence, it is difficult to characterize functional spaces
in which solution lies or to characterize quasi stationary solution, see Definition 1 in \cite{GomesIturriaga-2000}.
\par
In this paper we consider multidimensional Burgers equation \eqref{eq1} in $L^{p}(\mathcal{O},\mathbb{R}^d)$, $p>d$, in the domain $\mathcal O$ being either a torus $\mathbb T^d$ or the full space $\mathbb R^d$. In both cases we prove, in Theorems 4.1 and 4.3 respectively, the existence and uniqueness of solutions for every initial condition $u_0\in L^{p}(\mathcal{O},\mathbb{R}^d)$ and establish a priori estimates. In particular, Theorem 4.3 holds in the case $\mathcal O=\mathbb R^d$ and $\xi=0$ thus improving our previous results from \cite{Goldys-Neklyudov}. In the case of $\mathcal O=\mathbb R^d$ however, the a priori estimates are nonuniform with respect to $\nu$. Theorems 4.1 and 4.3 extend all aforementioned results on the existence and uniqueness of solutions to \eqref{eq1} to the stochastic case.\\
In Theorem 4.6 we provide a general sufficient condition under which uniform with respect to $\nu$ estimates can be derived on $\mathbb R^d$ as well. It is interesting to note that this condition can be viewed as a modification and an extension to the stochastic case of the famous Beale-Kato-Majda condition assuring the existence of global solutions to the deterministic Navier-Stokes equation.\\
Finally, we apply our results to the gradient case. It is easy to see that in the gradient case  the Beale-Kato-Majda condition holds and therefore the existence and uniqueness of global solutions follows from our general results. Morevoer, we obtain the estimates uniform in $\nu$ on the torus and on the whole space and as a consequence we show that there exists a vanishing viscosity limit for equation \eqref{eq1} for every $u_0\in L^p\left(\mathcal O,\mathbb R^d\right)$.
\par
In our proofs we extend the approach of \cite{Goldys-Neklyudov}, where the deterministic case $\xi=0$ was studied. We start with a proof of the local existence and uniqueness of mild solutions in $L^{p}(\mathcal{O},\mathbb{R}^d)$, $p\geq d$, following the argument of Weissler \cite{Weissler-1980}. Then we find a priori estimates using the Maximum Principle and then show that the local solution is in fact global. We note here that this method was applied earlier to the deterministic Burgers equation by Kiselev, Ladyzhenskaya \cite{[Kis+Lad_1957]}.
\par\medskip\noindent
\textbf{Acknowledgement}
We would like to thank Y. Sinai for pointing out reference \cite{[Kis+Lad_1957]}.

\section{Formulation of the problem and some auxiliary facts}
Let
$\mathcal{O}$ be either $\mathbb{T}^d$ or $\mathbb{R}^d$. In both cases we will use the same notation $\Delta$ for the generator of the heat semigroup $\left(S_t\right)$ in $\mathbb{L}^p(\mathcal{O}):=L^p\left(\mathcal O,\Rd\right)$ for $p\in(1,\infty)$. Let us recall that
\[
\mathrm{dom}_{\mathbb{L}^p(\mathcal{O})}(\Delta)=H^{2,p}\left(\Rd,\Rd\right)\quad\mbox{if}\quad \mathcal O=\Rd
\]
and
 \[
 \mathrm{dom}_{\mathbb{L}^p(\mathcal{O})}(\Delta)=H^{2,p}_{per}\left(\mathbb T^d,\Rd\right)\quad\mbox{if}\quad \mathcal O=\mathbb{T}^d.
 \]
We will use the standard notation $\mathbb H^{n,p}(\mathcal{O})=H^{n,p}\left(\mathcal{O},\Rd\right)$ for the Sobolev spaces of $\Rd$-valued functions with the norm
\[|f|_{n,p}=|(I-\triangle)^{\frac{n}{2}}f|_{\mathbb{L}^p(\mathcal{O})}.\]
The dual space space of $\mathbb H^{n,p}(\mathcal{O})$ will be denoted by $\mathbb H^{-n,q}(\mathcal{O})$ with $q=\frac{p}{p-1}$.
\par
Let $(\Omega,\mathcal{F},(\mathcal{F}_t)_{t\geq
0},\mathbb{P})$ be a probability space with the filtration satisfying the usual conditions. We will denote by $M^p([0,T],\mathbb H^{n,p}(\mathcal{O}))$ the space of $\mathbb H^{n,p}(\mathcal{O})$-valued progressively measurable processes endowed with the norm
\[\|u\|_{T,n,p}=\left(\mathbb{E}\intl_0^{T}|u(s,\cdot)|_{\mathbb H^{n,p}(\mathcal{O})}^p\,ds\right)^{\frac{1}{p}}.\]
Let $(W_t)_{t\geq 0}$ be a standard cylindrical Wiener process on separable Hilbert space $H$ defined on $(\Omega,\mathcal{F},(\mathcal{F}_t)_{t\geq
0},\mathbb{P})$. Let us recall that for $p\ge 2$ the space $\mathbb L^p(\mathcal{O})$ is an $M$-type 2 Banach space and therefore the stochastic integration theory is developed \cite{Brzezniak-1997} can be applied in this space. In order to give a meaning to equation \eqref{eq1} we will consider first its linearized version
\begin{equation}\label{eq2}
\left\{\begin{array}{ll}
\frac{\partial z}{\partial t}=\nu\Delta z+f+\xi,&t>0,\,\, x\in\mathcal O,\\
z(0,x)=0,&x\in\mathcal O.
\end{array}\right.
\end{equation}
that will be understood as a stochastic evolution equation in the space $\mathbb L^p(\mathcal O)$:
\begin{equation}
dz=(\nu\triangle
z+f)\,dt+g\,dW_t,\;\;z(0)=0.\label{eqn:StrongOrnsteinUhlenbeck-1}
\end{equation}
To define solution to equation \eqref{eqn:StrongOrnsteinUhlenbeck-1}, let us recall that for a Banach space $X$ and separable Hilbert space $H$, we denote by $\gamma(H,X)$ the Banach space of $\gamma$--radonifying operators from $H$ to $X$ (see definition 3.7 of \cite{JVNeerven2009}). If $g\in M^p\left([0,T];\gamma\left(H,\mathbb L^p(\mathcal O)\right)\right)$ and $f\in M^p\left([0,T];\mathbb H^{-1,p}(\mathcal O)\right)$ then solution to \eqref{eqn:StrongOrnsteinUhlenbeck-1} is given by the formula
\[z(t)=\int_0^tS_{t-s}^\nu f(s)ds+\int_0^tS_{t-s}^\nu g(s)dW(s).\]
The regularity properties of Ornstein-Uhlenbeck process were
studied in a vast number of articles, see for instance
\cite{Brzezniak-1997} and references therein. The following
theorem has been proved by Brzezniak \cite{Brzezniak-1997} (Corollary 3.5) and Krylov \cite{[Krylov-1999]}
(Theorem 4.10 (i) and Theorem 7.2(i), chapter 5) for the case of
whole space. The case of torus can be proved similarly.
\begin{theorem}\label{thm:OrnsteinUhlenbeckReg}
Assume $n\in\mathbb{Z}$ and $f\in M^p([0,T],\mathbb H^{n-1,p}(\mathcal{O}))$,
$g\in M^p([0,T],\gamma(H,\mathbb{H}^{n,p}(\mathcal{O})))$, $p>2$,
$\frac{1}{2}>\beta>\al>\frac{1}{p}$. Then equation
\eqref{eqn:StrongOrnsteinUhlenbeck-1} has unique solution $z\in
C^{\al-\frac{1}{p}}([0,T],\mathbb{H}^{n+1-2\beta,p}(\mathcal{O}))$ a.s..
\end{theorem}
For $\phi\in\mathbb H^{1,p}(\mathcal O)$ we define a function
\[
F(v)=(v\nabla)v.
\]

\begin{definition}\label{def-2.2}
Assume that 
$u_0\in
\mathbb L^p(\mathcal{O})$, $f\in M^p([0,T],\mathbb H^{-1,p}(\mathcal{O}))$,
$g\in
M^p([0,T],\gamma(H,\mathbb L^{p}(\mathcal{O})))$.
A progressively measurable $\mathbb L^p(\mathcal{O})$-valued continuous process $u$ defined on $[0,T]$ is said to be a mild solution of the
stochastic Burgers equation with the initial condition $u_0$ if $F(u(\cdot))\in L^1(0,T;\mathbb{L}^p(\mathcal{O}))$ a.s., $u=v+z$ where
$z:\Omega\to L^{\infty}(0,T;\mathbb{L}^p(\mathcal{O}))$
satisfies equation \eqref{eqn:StrongOrnsteinUhlenbeck-1} and
$v$ satisfies equality
\begin{equation}
v(t)=S_t^{\nu}u_0+\intl_0^tS_{t-s}^{\nu}(F(v(s)+z(s)))\, ds,\;t\in[0,T].\label{eqn:BurgersEqn-1}
\end{equation}
\del{\begin{equation}
dz=(\nu\triangle
z+f)dt+gdw_t,z(0)=0\label{eqn:OrnsteinUhlenbeck-1}
\end{equation}}
\end{definition}
\begin{remark}\label{rem-2.1}
We believe it is possible to define a weak solution to Burgers equation as in definition 8.5, p. 184 of \cite{BZ2010}.
Then it should be possible to prove that sufficiently regular process $u$ is a weak solution iff it is a mild solution, i.e. solves
\begin{equation}
u(t)=S_t^{\nu}u_0+\intl_0^tS_{t-s}^{\nu}(F(u(s)))\, ds+z(s),\;t\in[0,T],\label{eqn:BurgersEqn-1'}
\end{equation}
where $z$ satisfies equation \eqref{eqn:StrongOrnsteinUhlenbeck-1}.
In our paper we prove the existence and uniqueness of the solution of \eqref{eqn:BurgersEqn-1'}. This is done via a substitution
\begin{equation}
u=v+z.\label{eqn:Substitution_1}
\end{equation}
For a process of the form \eqref{eqn:Substitution_1} we can prove that it is a mild solution iff it is a strong solution according to
the following definition.
\end{remark}
\begin{definition}\label{def-2.3}
Assume that  $u_0$, $f$,
$g$ satisfy the same assumptions as in the definition \ref{def-2.2}.
We call progressively measurable process $u:\Omega\to
L^{\infty}(0,T;\mathbb{L}^p(\mathcal{O}))$ a strong solution of
stochastic Burgers equation with the initial condition $u_0$  
iff $F(u)\in L^1(0,T;\mathbb{L}^p(\mathcal{O}))$ a.s. and $u=v+z$ where
$z:\Omega\to L^{\infty}(0,T;\mathbb{L}^p(\mathcal{O}))$
satisfies equation \eqref{eqn:StrongOrnsteinUhlenbeck-1} and
$v\in C^1((0,T];\mathbb{L}^p(\mathcal{O}))$ satisfies equality
\begin{eqnarray}
\frac{\partial v}{\partial t}(t)&=&\nu\triangle
v(t)+F(v(t)+z(t)),\;t\in[0,T]\label{eqn:StrongBurgersEqn-1}\\
v(0)&=&u_0.
\end{eqnarray}
\end{definition}
\begin{remark}
It is possible to define in a similar fashion strong and mild solution of stochastic Burgers equation without referring to the Ornstein-Uhlenbeck process $z$. However, the definition given above has certain merit since it allows to transfer all noise effects to the process $z$ and consider
PDE with random coefficients instead of SPDE.
\end{remark}

\section{The existence of a local solution to the Stochastic Burgers equation}

Theorem \ref{thm:OrnsteinUhlenbeckReg} allow us to work pathwise
i.e. we assume that some version of $z$ of specified regularity is
fixed.

Local existence of solution of Burgers equation in
$\mathbb{L}^p(\mathcal{O})$ can be shown in the same way as for
Navier-Stokes equation (see
\cite{FujitaKato-1964}, \cite{FujitaKato-1962}, \cite{Kato-1984}, \cite{Weissler-1980}, \cite{FabesJonesRiviere-1972}
and others). Here we only state main points of the proof following
the work of Weissler \cite{Weissler-1980}.

We will use following version of abstract theorem proved in
\cite{Weissler-1980}, p. 222, Theorem 2, see also
\cite{FujitaKato-1962} and \cite{Kato-1984}.

\begin{theorem}\label{thm:AbstractExistence}
Let $W$, $X$, $Y$, $Z$ be Banach spaces continuously embedded in
some topological vector space $\mathcal{X}$. $R_t=e^{tA},\;t\geq 0$
be $C_0$-semigroup on X, which satisfies the following additional
conditions
\begin{trivlist}
\item[(a1)] For each $t>0$, $R_t$ extends to a bounded map $W\to
X$. For some $a>0$ there are positive constants $C$ and $T$ such
that
\begin{equation}
|R_th|_X\leq Ct^{-a}|h|_W,h\in W,\;t\in(0,T].\label{eqn:a1Cond}
\end{equation}
\item[(a2)]For each $t>0$, $R_t$ is a bounded map $X\to Y$. For
some $b>0$ there are positive constants $C$ and $T$ such that
\begin{equation}
|R_th|_Y\leq Ct^{-b}|h|_X,h\in X,\;t\in(0,T].\label{eqn:a2Cond}
\end{equation}
Furthermore, function $|R_th|_Y\in C((0,T]), h\in X$ and
\begin{equation}
\liml_{t\to 0+}t^b|R_th|_Y=0,\forall h\in X.\label{eqn:a2Cond'}
\end{equation}

\item[(a3)]For each $t>0$, $R_t$ is a bounded map $X\to Z$. For
some $c>0$ there are positive constants $C$ and $T$ such that
\begin{equation}
|R_th|_Z\leq Ct^{-c}|h|_X,h\in X,\;t\in(0,T].\label{eqn:a3Cond}
\end{equation}
Furthermore, function $|R_th|_Z\in C((0,T]), h\in X$ and
\begin{equation}
\liml_{t\to 0+}t^c|R_th|_Z=0,\forall h\in X.\label{eqn:a3Cond'}
\end{equation}
\end{trivlist}
Let also $G:Y\times Z\to W$ be a bounded bilinear map, $L\in
L^{\infty}(0,T;\mathcal{L}(Y\cap Z,W))$, and let $G(u)=G(u,u),u\in
Y\cap Z$, $f\in L^{\infty}(0,T;W)$. Assume also that $a+b+c\leq
1$.

Then for each $u_0\in X$ there is $T>0$ and unique function
$u:[0,T]\to X$ such that:
\begin{trivlist}
\item[(a)] $u\in C([0,T],X)$, $u(0)=u_0$. \item[(b)] $u\in
C((0,T],Y)$, $\liml_{t\to 0+}t^b|u(t)|_Y=0$. \item[(c)] $u\in
C((0,T],Z)$, $\liml_{t\to 0+}t^c|u(t)|_Z=0$.\item[(d)]
$$
u(t)=R_tu_0+\intl_0^tR_{t-\tau}(G(u(\tau))+L(u(\tau))+f(\tau))d\tau,\;t\in
[0,T]
$$
\end{trivlist}
\end{theorem}
\begin{remark}
Weissler \cite{Weissler-1980} considers only the case of $L=f=0$.
The general case follows similarly (see also \cite{Goldys-Neklyudov}).
\end{remark}
In the next proposition we will summarize properties of the heat
semigroup on $\mathcal{O}$.
\begin{proposition}\label{prop:HeatSemigroup}
Assume that either $\mathcal{O}=\mathbb{T}^d$ or $\mathcal{O}=\mathbb{R}^d$ and $\triangle$ is a corresponding periodic (respectively free) Laplacian with domain of definition $\mathrm{dom}_{\mathbb{L}^p(\mathcal{O})}(\Delta)$.
Then
\begin{trivlist}
\item[(i)]
\begin{eqnarray}
|\nabla^me^{t\triangle}h|_{\mathbb{L}^q(\mathcal{O})}\leq c
t^{-\frac{m}{2}-\frac{d}{2r}}|h|_{\mathbb{L}^p(\mathcal{O})},\;t\in (0,T],\label{eqn:HeatSemigroupEstimate-1}\\
\frac{1}{r}=\frac{1}{p}-\frac{1}{q},\quad 1<p\leq q<\infty,\quad
h\in \mathbb{L}^p(\mathcal{O}).\nonumber
\end{eqnarray}
Furthermore,
\begin{equation}
\liml_{t\to
0+}t^{\frac{m}{2}+\frac{d}{2r}}|\nabla^me^{t\triangle}h|_{\mathbb{L}^q(\mathcal{O})}=0,
\quad h\in \mathbb{L}^p(\mathcal{O}).\label{eqn:HeatSemigroupLimit-1}
\end{equation}
\item[(ii)] Let $p\in(1,\infty)$. Then for any $t>0$,
$e^{t\triangle}:\mathbb{L}^p(\mathcal{O})\to \mathbb{H}^{1,p}(\mathcal{O})$ is a bounded map. Moreover, for
each $T>0$ there exists $C=C(p,T)$, such that
\begin{equation}
|e^{t\triangle}h|_{\mathbb{H}^{1,p}(\mathcal{O})}\leq Ct^{-\frac{1}{2}}|h|_{\mathbb{L}^p(\mathcal{O})},\;t\in
(0,T],h\in \mathbb{L}^p(\mathcal{O}).\label{eqn:HeatSemigroupEstimate-3}
\end{equation}
Furthermore,
\begin{equation}
\liml_{t\to 0+}t^{\frac{1}{2}}|e^{t\triangle}h|_{\mathbb{H}^{1,p}(\mathcal{O})}=0,\quad
h\in \mathbb{L}^p(\mathcal{O}).\label{eqn:HeatSemigroupLimit-3}
\end{equation}
\end{trivlist}
\end{proposition}
\begin{proof}
See for example books by 
Lunardi \cite{Lunardi} or by Quittner, Souplet \cite{QuittnerSouplet}. 
\end{proof}
Now we can formulate the following results about the existence and uniqueness of a local mild solution of the auxiliary deterministic problem.
\begin{theorem}\label{thm:LocalExistence-2} Assume that $p\geq d$. Then
for all $u_0\in \mathbb{L}^p(\mathcal{O})$, $z\in
L^{\infty}(0,T;\mathbb{L}^{2p}(\mathcal{O})\cap \mathbb{H}^{1,p}(\mathcal{O}))$,
 there exists
$T_0=T_0(\nu,|u_0|_{\mathbb{L}^p(\mathcal{O})},|z|_{L^{\infty}(0,T;\mathbb{L}^{2p}(\mathcal{O})\cap
\mathbb{H}^{1,p}(\mathcal{O}))})>0$ such that there exists unique mild
solution $u\in L^{\infty}(0,T_0;\mathbb{L}^p(\mathcal{O}))$ of
equation \eqref{eqn:BurgersEqn-1}. Furthermore
\begin{trivlist}
\item[(a)]$u:(0,T_0]\to
\mathbb{L}^{2p}(\mathcal{O})$ is continuous and $\liml_{t\to
0}t^{\frac{d}{4p}}|u(t)|_{\mathbb{L}^{2p}(\mathcal{O})}=0$. \item[(b)]$u:(0,T_0]\to
\mathbb{H}^{1,p}(\mathcal{O})$ is continuous and $\liml_{t\to
0}t^{\frac{1}{2}}|u(t)|_{\mathbb{H}^{1,p}(\mathcal{O})}=0$.
\end{trivlist}
\end{theorem}

\begin{proof}[Proof of Theorem \ref{thm:LocalExistence-2}]
We apply Theorem \ref{thm:AbstractExistence} to equation
\eqref{eqn:BurgersEqn-1} with $X=\mathbb{L}^p(\mathcal{O})$, $Y=\mathbb{L}^{2p}(\mathcal{O})$, $Z=\mathbb{H}^{1,p}(\mathcal{O})$,
$W=\mathbb{L}^{\frac{2p}{3}}(\mathcal{O})$, $R_{\cdot}=e^{\cdot\triangle}$, $L=F(z,\cdot)+F(\cdot,z)$, $f=F(z,z)$. It follows from the H\"{o}lder inequality that $F:\mathbb{L}^{2p}(\mathcal{O})\times
\mathbb{H}^{1,p}(\mathcal{O})\to \mathbb{L}^{\frac{2p}{3}}(\mathcal{O})$ is a bounded bilinear map. Hence the function $f$ satisfies the assumption of Theorem \eqref{thm:AbstractExistence}, i.e. $f\in L^\infty(0,T;W)$.
Condition
\eqref{eqn:a1Cond} is satisfied with $a=\frac{d}{4p}$ by estimate
\eqref{eqn:HeatSemigroupEstimate-1}. Conditions
\eqref{eqn:a2Cond}, \eqref{eqn:a2Cond'} are satisfied with
$b=\frac{d}{4p}$ by \eqref{eqn:HeatSemigroupEstimate-1} and
\eqref{eqn:HeatSemigroupLimit-1}. Conditions \eqref{eqn:a3Cond},
\eqref{eqn:a3Cond'} are satisfied with $c=\frac{1}{2}$ by
\eqref{eqn:HeatSemigroupEstimate-3} and
\eqref{eqn:HeatSemigroupLimit-3}.
\end{proof}
\begin{corollary}\label{cor:ClassicalSolutionReg}
Let $p\geq d$, $\theta\in (0,1)$, $u_0\in
\mathbb{L}^p(\mathcal{O})$, $z\in L^{\infty}(0,T;\mathbb{H}^{1,2p}(\mathcal{O})\cap
\mathbb{H}^{1,p}(\mathcal{O}))$, $z\in C^{\theta}((0,T],\mathbb{H}^{1,2p}(\mathcal{O}))$.
Then $v\in C^1((0,T];\mathbb{L}^p(\mathcal{O}))\cap C((0,T];\mathbb{H}^{2,p}(\mathcal{O}))\cap
C^{\theta}_{\textrm{loc}}((0,T],\mathbb{H}^{2,p}(\mathcal{O}))\cap C_{\textrm{loc}}^{1+\theta}((0,T],\mathbb{L}^p(\mathcal{O}))$
and $v$ is a solution to the  system
\begin{equation}\label{eqn:BurgersEquation-11}
v'=\nu\triangle v-F(v+z).
\end{equation}
\end{corollary}

\begin{proof}[Proof of Corollary \ref{cor:ClassicalSolutionReg}]

Let us show that there exist $T_1$ such that $v\in
C((0,T_1],\mathbb{H}^{1,2p}(\mathcal{O}))$ and $\liml_{t\to
0}t^{\frac{1}{2}}|u(t)|_{\mathbb{H}^{1,2p}(\mathcal{O})}=0$. We apply
Theorem \ref{thm:AbstractExistence} with following data
$X=Y=\mathbb{L}^p(\mathcal{O})$, $Z=\mathbb{H}^{1,2p}(\mathcal{O})$,
$W=\mathbb{L}^{\frac{2p}{3}}(\mathcal{O})$. Then it follows from
H\"{o}lder inequality that $F:\mathbb{L}^p(\mathcal{O})\times
\mathbb{H}^{1,2p}(\mathcal{O})\to \mathbb{L}^{\frac{2p}{3}}(\mathcal{O})$ is
a bounded bilinear map. Conditions \eqref{eqn:a1Cond} is satisfied
with $a=\frac{d}{4p}$ by estimate
\eqref{eqn:HeatSemigroupEstimate-1}. Conditions
\eqref{eqn:a2Cond},\eqref{eqn:a2Cond'} are satisfied with
arbitrary $b>0$ because heat semigroup is analytic on
$\mathbb{L}^p(\mathcal{O})$. Conditions
\eqref{eqn:a3Cond},\eqref{eqn:a3Cond'} are satisfied with
$c=\frac{1}{2}$ by \eqref{eqn:HeatSemigroupEstimate-3} and
\eqref{eqn:HeatSemigroupLimit-3}.

As a result by part (c) of the Theorem \ref{thm:AbstractExistence}
we get existence of $T_1$ such that $v\in
C((0,T_1],\mathbb{H}^{1,2p}(\mathcal{O}))$ and $\liml_{t\to
0}t^{\frac{1}{2}}|v(t)|_{\mathbb{H}^{1,2p}(\mathcal{O})}=0$. Put
$T_2=\min\{T,T_1\}$.
Therefore, we have
\begin{align}
|F(v+z)|_{L^1(0,T_2;\mathbb{L}^p(\mathcal{O}))}&\leq\intl_0^{T_2}|v(s)+z(s)|_{\mathbb{L}^{2p}(\mathcal{O})}|\nabla
(v+z)|_{\mathbb{L}^{2p}(\mathcal{O})}\,ds\nonumber\\
&\leq\intl_0^{T_2}\frac{1}{s^{\frac{d}{4p}+\frac{1}{2}}}\supl_s(s^{\frac{d}{4p}}|v(s)+z(s)|_{\mathbb{L}^{2p}})\supl_s(s^{\frac{1}{2}}|v(s)+z(s)|_{\mathbb{H}^{1,2p}})ds\nonumber\\
&\leq\supl_s(s^{\frac{d}{4p}}|v(s)+z(s)|_{\mathbb{L}^{2p}})\supl_s(s^{\frac{1}{2}}|v(s)+z(s)|_{\mathbb{H}^{1,2p}})T_2^{\frac{1}{2}-\frac{d}{4p}}<\infty.\label{eqn:aux-a}
\end{align}
Let us show that for any $\eps>0$ the function $F(v(\cdot)+z(\cdot)):[\eps,T_2]\to \mathbb{L}^p(\mathcal{O})$ is H\"{o}lder
continuous of order $\left(\frac{1}{2}-\frac{d}{4p}\right)$. Then the result will follow from
Theorem 4.3.4, p. 137 in \cite{Lunardi} and inequality \eqref{eqn:aux-a}. Since
$F:\mathbb{H}^{1,2p}(\mathcal{O})\to \mathbb{L}^p(\mathcal{O})$ is locally Lipschitz it is easy to notice that
it is enough to prove that $v:[\eps,T_2]\to \mathbb{H}^{1,2p}(\mathcal{O})$ is
H\"{o}lder continuous for any $\eps>0$. Since we have
the representation
\begin{equation}
v(t)=S_{t-\eps}v(\eps)-\intl_{\eps}^tS_{t-s}(F(v(s)+z(s)))ds,\;t\in
[\eps,T_2],\label{eqn:BurgersEquation-3}
\end{equation}
it is enough to show that each term of this representation
is H\"{o}lder continuous. Similarly to \eqref{eqn:aux-a} we have
\begin{equation}
\supl_{t\in[0,T_2]}t^{\frac{1}{2}+\frac{d}{4p}}|F(v(t)+z(t))|_{\mathbb{L}^p}\leq
\supl_ss^{\frac{d}{4p}}|v(s)+z(s)|_{\mathbb{L}^{2p}}\supl_ss^{\frac{1}{2}}|v(s)+z(s)|_{\mathbb{H}^{1,2p}}<\infty
\end{equation}
and it follows by Proposition 4.2.3 part (i), p.130 of
\cite{Lunardi} that $\intl_0^tS_{t-s}F(v(s)+z(s))\,ds\in
C^{\frac{1}{2}-\frac{d}{4p}}(0,T_2;L^p)$.
\end{proof}
\begin{corollary}\label{cor:ClassicalSolutionReg-2}
Suppose that assumptions of Corollary
\ref{cor:ClassicalSolutionReg} are satisfied. Assume also that
$z\in C^{\theta}((0,T],\mathbb{H}^{k+1,2p}(\mathcal{O}))$,
 for some $k\in \Nat$. Then $v\in
C^{\theta}((0,T],\mathbb{H}^{k+2,p}(\mathcal{O}))\cap
C^{1+\theta}((0,T],\mathbb{H}^{k,p}(\mathcal{O}))$.
\end{corollary}
\begin{proof}
We will show the result for $k=1$. General case follows similarly. We fix some $\eps>0$.
We have by the Theorem \ref{thm:LocalExistence-2}, part (a) that $v(\eps)\in \mathbb{L}^{2p}(\mathcal{O})$.
 As a result, by means of the corollary \ref{cor:ClassicalSolutionReg} we infer that
\begin{equation}
v\in C^{\theta}([\eps,T],\mathbb{H}^{2,2p}(\mathcal{O}))\cap
C^{1+\theta}([\eps,T],\mathbb{L}^{2p}(\mathcal{O})).
\label{eqn:ClassicalSolutionReg-3}
\end{equation}
Hence,
\begin{equation}
v+z\in
C^{\theta}([\eps,T],\mathbb{H}^{2,2p}(\mathcal{O}))\label{eqn:ClassicalSolutionReg-3'}
\end{equation}
Therefore, we have following estimates for nonlinearity
\begin{align}
|F(v+z)|_{C^{\theta}([\eps,T],\mathbb{L}^p(\mathcal{O}))}&\leq
|v+z|_{L^{\infty}(\eps,T;\mathbb{L}^{2p}(\mathcal{O}))}|\nabla(v+z)|_{C^{\theta}([\eps,T],\mathbb{L}^{2p}(\mathcal{O}))}\nonumber\\
&+|\nabla
(v+z)|_{L^{\infty}(\eps,T;\mathbb{L}^{2p}(\mathcal{O}))}|v+z|_{C^{\theta}([\eps,T],\mathbb{L}^{2p}(\mathcal{O}))}<\infty\label{eqn:ClassicalSolutionReg-4}
\end{align}
where we have used \eqref{eqn:ClassicalSolutionReg-3'}.
Furthermore,
\begin{align}
|\nabla F(v+z)|_{C^{\theta}([\eps,T],\mathbb{L}^p(\mathcal{O}))}&\leq C |\nabla
(v+z)|_{L^{\infty}(\eps,T;\mathbb{L}^{2p}(\mathcal{O}))}|\nabla
(v+z)|_{C^{\theta}([\eps,T],\mathbb{L}^{2p}(\mathcal{O}))}\nonumber\\
&+|v+z|_{L^{\infty}(\eps,T;\mathbb{L}^{2p}(\mathcal{O}))}|\triangle
(v+z)|_{C^{\theta}([\eps,T],\mathbb{L}^{2p}(\mathcal{O}))}\nonumber\\
&+|\triangle
(v+z)|_{L^{\infty}(\eps,T;\mathbb{L}^{2p}(\mathcal{O}))}|v+z|_{C^{\theta}([\eps,T],\mathbb{L}^{2p}(\mathcal{O}))}<\infty,\label{eqn:ClassicalSolutionReg-5}
\end{align}
where we have used \eqref{eqn:ClassicalSolutionReg-3'}. Thus,
combining inequalities \eqref{eqn:ClassicalSolutionReg-4} and
\eqref{eqn:ClassicalSolutionReg-5} we get $F(v+z)\in
C^{\theta}([\eps,T],\mathbb{H}^{1,p}(\mathcal{O}))$, $\forall \eps>0$.  Therefore by
a maximal regularity result, see Theorem 4.3.1, p. 134 of \cite{Lunardi},
it follows that $v\in C^{\theta}([\eps,T],\mathbb{H}^{3,p}(\mathcal{O}))\cap
C^{1+\theta}([\eps,T],\mathbb{H}^{1,p}(\mathcal{O}))$.
\end{proof}

In the next lemma we will show that either local solution defined
in previous theorems is global or it blows up. Let us denote by
$T_{max}$ maximal existence time of solution.
\begin{lemma}\label{lem:LocalBlowUpBehaviour}
Assume that $u_0\in \mathbb{L}^p(\mathcal{O})$, $z\in
L^{\infty}([0,T],\mathbb{H}^{1,2p}(\mathcal{O})\cap \mathbb{H}^{1,p}(\mathcal{O}))$, $z\in
C^{\theta}(0,T;\mathbb{L}^p(\mathcal{O}))$, $p>d$, $\theta\in (0,1)$ and
$T_{max}<T$. Let $u\in
{\blue C}([0,T_{max});\mathbb{L}^p(\mathcal{O}))$ be a maximal
local mild solution to the Burgers equation \eqref{eqn:BurgersEqn-1}. Then
\begin{equation}
\limsup\limits_{t\nearrow
T_{max}}|u(t)|_{\mathbb{L}^p(\mathcal{O})}^2=\infty.\label{eqn:LocalBlowUpBehaviour-1}
\end{equation}
\end{lemma}
\begin{question}
Can the Lemma \ref{lem:LocalBlowUpBehaviour} be strengthened to show
$\lim$ instead of $\limsup$ in the equality \eqref{eqn:LocalBlowUpBehaviour-1}?
\end{question}
\begin{question}
It would be interesting to extend the Lemma
\ref{lem:LocalBlowUpBehaviour} to the case when $p=d$.
\end{question}
\begin{proof}[Proof of Lemma \ref{lem:LocalBlowUpBehaviour}]
We will argue by contradiction. Assume that there exits $R>0$ such that
\begin{equation}
|u(t)|_{\mathbb{L}^p(\mathcal{O})}^2 \leq R, \;\;t\in (0,T_{max}].\label{eqn:LocalBlowUpBehaviour-2}
\end{equation}
Let us denote
\begin{equation}
K_1=\supl_{t\in
[0,T_{max})}|u(t)|_{\mathbb{L}^{p}(\mathcal{O})}<\infty.\label{eqn:LocalLpBound}
\end{equation}
Since $z\in L^{\infty}([0,T],\mathbb{H}^{1,p}(\mathcal{O}))$ we have the similar bound for
$v=u-z$:
\begin{equation}
\supl_{t\in
[0,T_{max})}|v(t)|_{\mathbb{L}^{p}(\mathcal{O})}<K_1'=K_1+|z|_{L^{\infty}([0,T_{max}],\mathbb{H}^{1,p}(\mathcal{O}))}.\label{eqn:LocalLpBound-1}
\end{equation}
Let us fix $\eps\in (0,T_{max})$. We will show that there exist $C,\al>0$ such that
\begin{equation}
|u(t)-u(\tau)|_{\mathbb{L}^p(\mathcal{O})}\leq C|t-\tau|^{\al},\;t,\tau\in
[\eps,T_{max}).\label{eqn:UniformContinuity-1}
\end{equation}
Then it follows from inequalities \eqref{eqn:LocalBlowUpBehaviour-2} and
\eqref{eqn:UniformContinuity-1} 
that there exist $y\in \mathbb{L}^p(\mathcal{O})$ such that
\begin{equation}
\liml_{t\nearrow T_{max}}
|u(t)-y|_{\mathbb{L}^p(\mathcal{O})}=0,\label{eqn:LocalBlowUpBehaviour-3}
\end{equation}
and we have a contradiction with the definition of $T_{max}$. Thus, we
need to show the inequality \eqref{eqn:UniformContinuity-1}. We will first show that
\begin{equation}
K_2=\supl_{t\in
[\eps,T_{max})}|u(t)|_{\mathbb{H}^{1,p}(\mathcal{O})}<\infty.\label{eqn:LocalH1pBound}
\end{equation}
It is enough to show
\begin{equation}
\supl_{t\in [\eps,T_{max})}|\nabla
v(t)|_{\mathbb{L}^{p}(\mathcal{O})}<\infty.\label{eqn:LocalH1pBound-2}
\end{equation}
Indeed, the inequality \eqref{eqn:LocalH1pBound}
immediately follows from inequalities \eqref{eqn:LocalLpBound-1},
\eqref{eqn:LocalH1pBound-2} and the regularity of $z$. We have
\begin{equation}
\nabla v(t)=\nabla S_{t-\eps}v(\eps)-\intl_{\eps}^t\nabla
S_{t-s}(F(v(s)+z(s)))\,ds.
\end{equation}
Hence, for $t\in (\eps,T_{max})$ we have
\begin{eqnarray}
|\nabla v(t)|_{\mathbb{L}^{p}(\mathcal{O})}&\leq& |\nabla
S_{t-\eps}v(\eps)|_{\mathbb{L}^p(\mathcal{O})}+\intl_{\eps}^t|\nabla
S_{t-s}F(v(s)+z(s))|_{\mathbb{L}^p(\mathcal{O})}\,ds\nonumber\\
&\leq& \frac{C|v(\eps)|_{\mathbb{L}^p(\mathcal{O})}}{(t-\eps)^{1/2}}
+C\intl_{\eps}^t\frac{|S_{(t-s)/2}F(v(s)+z(s))|_{\mathbb{L}^p(\mathcal{O})}}{|t-s|^{1/2}}\,ds\nonumber\\
&\leq& \frac{C|u_0|_{\mathbb{L}^p(\mathcal{O})}}{t^{1/2}}
+C\intl_{\eps}^t\frac{|F(v(s)+z(s))|_{\mathbb{L}^{p/2}(\mathcal{O})}}{|t-s|^{1/2+d/(2p)}}\,ds\nonumber\\
&\leq& \frac{C|u_0|_{\mathbb{L}^p(\mathcal{O})}}{t^{1/2}} +C\intl_{\eps}^t
\frac{|v(s)+z(s)|_{\mathbb{L}^{p}(\mathcal{O})}}{|t-s|^{1/2+d/(2p)}}|\nabla
v(s)+\nabla z(s)|_{\mathbb{L}^{p}(\mathcal{O})}\,ds\nonumber\\
&\leq&\frac{C|u_0|_{\mathbb{L}^p(\mathcal{O})}}{t^{1/2}}+C(K_1')^2T
+CK_1'\intl_{\eps}^t\frac{|\nabla
v(s)|_{\mathbb{L}^{p}(\mathcal{O})}}{|t-s|^{1/2+d/(2p)}}\,ds<\infty,
\end{eqnarray}
where the second and third inequalities follow from the property
\eqref{eqn:HeatSemigroupEstimate-1} of the heat semigroup, fourth inequality follows from the
H\"older inequality, assumption \eqref{eqn:LocalLpBound} is used
in the fifth one and the last inequality follows from (c), Theorem
\ref{thm:LocalExistence-2}. Now if $\frac{1}{2}+\frac{d}{2p}<1$
(i.e. if $p>d$) we can use a version of the Gronwall inequality (\cite{Henry-1981},
Lemma 7.1.1, p. 188) to conclude that estimate
\eqref{eqn:LocalH1pBound-2} holds. Thus we get an estimate
\eqref{eqn:LocalH1pBound}.

Now we can turn to the proof of uniform continuity condition \eqref{eqn:UniformContinuity-1}.
Since $u=v+z$ and $z\in C^{\theta}(0,T;\mathbb{L}^p(\mathcal{O}))$ it is enough to show
\eqref{eqn:UniformContinuity-1} with $v$ instead of $u$. We have
\begin{equation}
v(t)-v(\tau)=S_{t-\tau}v(\tau)-v(\tau)-\intl_{\tau}^tS_{t-s}(F(v(s)+z(s)))\,ds.
\end{equation}
Then
\begin{eqnarray}
|v(t)-v(\tau)|_{\mathbb{L}^p(\mathcal{O})}\leq
|S_{t-\tau}v(\tau)-v(\tau)|_{\mathbb{L}^p(\mathcal{O})}\nonumber\\
+|\intl_{\tau}^tS_{t-s}F(v(s)+z(s))\,ds|_{\mathbb{L}^p(\mathcal{O})}=(I)+(II).
\end{eqnarray}
The first term can be estimated as follows, where the $\sup $ is taken over the set $\{\phi\in
C_0^{\infty}(\mathcal{O}):|\phi|_{\mathbb{L}^q(\mathcal{O})}=1\}$,
\begin{eqnarray}\nonumber
(I)&=&\supl|\lb S_{t-\tau}v(\tau)-v(\tau),\phi \rb|
=\supl|\lb v(\tau),S_{t-\tau}\phi-\phi\rb |\\
&=&\supl|\lb v(\tau),\nu\intl_{\tau}^t\triangle
S_{s-\tau}\phi
ds\rb |
=\nu\supl|\lb \nabla
v(\tau),\intl_{\tau}^t\nabla S_{s-\tau}\phi ds\rb |\nonumber\\
&\leq&\nu\supl|\nabla
v(\tau)|_{\mathbb{L}^p(\mathcal{O})}\intl_{\tau}^t|\nabla S_{s-\tau}\phi|_{\mathbb{L}^q(\mathcal{O})}\,ds\nonumber\\
&\leq&\nu\supl|\nabla
v(\tau)|_{\mathbb{L}^p(\mathcal{O})}C\intl_{\tau}^t\frac{|\phi|_{\mathbb{L}^q(\mathcal{O})}}{|s-\tau|^{1/2}}\,ds
\leq\nu CK_2|t-\tau|^{1/2}.\label{eqn:term1Est}
\end{eqnarray}
For the second term, by using  the property
\eqref{eqn:HeatSemigroupEstimate-1} of heat semigroup and the H\"older inequality we have
\begin{eqnarray}
(II)\leq\intl_{\tau}^t\frac{|F(v(s)+z(s))|_{\mathbb{L}^{p/2}(\mathcal{O})}}{|t-s|^{\frac{d}{2p}}}\,ds\leq
\intl_{\tau}^t\frac{|u(s)|_{\mathbb{L}^p(\mathcal{O})}|\nabla
u(s)|_{\mathbb{L}^p(\mathcal{O})}}{|t-s|^{\frac{d}{2p}}}\,ds\nonumber\\
\leq CK_2^2|t-\tau|^{1-\frac{d}{2p}}.\label{eqn:term3Est}
\end{eqnarray}
Finally,  the last inequality follows from estimate
\eqref{eqn:LocalH1pBound}.

Combining inequalities \eqref{eqn:term1Est} and \eqref{eqn:term3Est} we get
\eqref{eqn:UniformContinuity-1}.
\end{proof}

\section{The existence of a global solution to Stochastic Burgers equation}
In this section we continue to work pathwise. We will now state the global
existence result for the case of torus.
\begin{theorem}\label{thm:TorusGlobalExistence}
Fix $p>d$. Assume that $u_0\in \mathbb{L}^p(\mathbb{T}^d)$ a.s.,
$f\in M^{2p}([0,T],\mathbb H^{3,2p}(\mathbb{T}^d))$, and
$g\in M^{2p}([0,T],\gamma(H,\mathbb H^{4,2p}(\mathbb{T}^d)))$. Then there exists a unique strong global $\mathbb{L}^p(\mathbb{T}^d)$-valued solution $u$ of the
Burgers equation. Moreover,
\begin{equation}
|u(t)|_{\mathbb{L}^p(\mathbb{T}^d)}^p\leq
C(|u_0|_{\mathbb{L}^p(\mathbb{T}^d)}^p+|z|_{L^{\infty}(0,T;\mathbb{H}^{2,p}(\mathbb{T}^d))}^2)e^{|\nabla
z|_{L^1(0,T;\mathbb{L}^{\infty}(\mathbb{T}^d))}},t\in [0,T].\label{eqn:TEnergyInequality}
\end{equation}
\end{theorem}

\begin{proof}[Proof of Theorem \ref{thm:TorusGlobalExistence}]
Suppose $T_{max}<T$. Let $v=u-z$. It is enough to find an
estimate for $v$ in $\mathbb{L}^{\infty}(\mathbb{T}^d)$ norm. Then the estimate
of norm of $v$ in $\mathbb{L}^p(\mathbb{T}^d)$ will immediately follow. Therefore, it is enough to prove that for any fixed $0<\delta<T_{max}$,  we have
\begin{equation}
|v|_{L^{\infty}((0,T_{max};\mathbb{L}^{\infty}(\mathbb{T}^d))}\leq
(|v(\delta)|_{\mathbb{H}^{1,p}(\mathbb{T}^d)}+|z|_{L^{\infty}(0,T_{max};\mathbb{H}^{2,p}(\mathbb{T}^d))}^2)e^{T_{max}+|\nabla
z|_{L^1(0,T_{max};\mathbb{L}^{\infty}(\mathbb{T}^d))}}.\label{FeynmanKacEst}
\end{equation}
to prove \eqref{FeynmanKacEst} we note first that the local solution $v$ satisfies the equation
\[
v'=\nu\triangle v-(v+z)\nabla v-v\nabla z-(z\nabla)z.
\]
Let $\phi(t)=v e^{-\int\limits_0^{t}(1+|\nabla z|_{\mathbb{L}^\infty})\,ds}-|z|_{L^{\infty}(0,T_{max};\mathbb{H}^{2,p}(\mathbb{T}^d))}^2, t\in [0,T_{max})$. Then
\[
\phi'=\nu\triangle \phi-(v+z)\nabla \phi+(\phi+|z|_{L^{\infty}(0,T_{max};\mathbb{H}^{2,p}(\mathbb{T}^d))}^2)(-\nabla z-|\nabla z|_{\mathbb{L}^\infty}-1)-(z\nabla)ze^{-\int\limits_0^{\cdot}(1+|\nabla z|_{\mathbb{L}^\infty})\,ds}
\]
or, equivalently,
\begin{align*}
\nu\triangle \phi&-(v+z)\nabla \phi+\phi(-\nabla z-|\nabla z|_{\mathbb{L}^\infty}-1)-\phi'\\
&=|z|_{L^{\infty}(0,T_{max};\mathbb{H}^{2,p}(\mathbb{T}^d))}^2(\nabla z+|\nabla z|_{\mathbb{L}^\infty}+1)+(z\nabla)ze^{-\int\limits_0^{\cdot}(1+|\nabla z|_{\mathbb{L}^\infty})\,ds}\geq 0.
\end{align*}
Now \eqref{FeynmanKacEst} follows from the Maximum Principle (Theorem 7, p. 174, \cite{ProtterWeinberger}).
\end{proof}
\begin{remark}
We  remark that in the a priori estimate above we can take the limit $\nu\to 0$ under appropriate
assumptions for the noise.
\end{remark}
Now we will formulate similar results for the whole space. In this
case we do not have an embedding $\mathbb{L}^{\infty}\subset \mathbb{L}^p$. Hence
estimate in $\mathbb{L}^{\infty}$ does not imply estimate in $\mathbb{L}^p$.
Nevertheless it is possible to get the following global existence result.
\begin{theorem}\label{thm:GlobalExistence}
Fix $p>d$. Assume that $u_0\in \mathbb{L}^p(\Rd)$ a.s.,
$f\in M^{2p}([0,T],\mathbb H^{3,2p}(\Rd))$,
$g\in M^{2p}([0,T],\gamma(H,\mathbb H^{4,2p}(\Rd)))$. Then there exists a unique strong global $\mathbb{L}^p(\mathbb{R}^d)$-valued solution $u$ of
Burgers equation and we have
\begin{align}
|u(t)|_{\mathbb{L}^p(\Rd)}^p&\leq
(|u_0|_{\mathbb{L}^p(\Rd)}^p+|F(z)|_{L^1(0,t;\mathbb{L}^p(\Rd))})\exp\Big\{(p+1)t|z|_{L^{\infty}(0,t;\mathbb{H}^{2,p}(\Rd))}+(p-1)t\nonumber\\
&+\frac{2t}{\nu p}(C|u_0|_{\mathbb{L}^p(\Rd)}^{2p}+|z|_{L^{\infty}(0,t;\mathbb{H}^{2,p}(\Rd))}^4)e^{2t|z|_{L^{\infty}(0,t;\mathbb{H}^{2,p}(\Rd))}}\Big\},\;t\geq 0.\label{eqn:EnergyInequality-2}
\end{align}
\end{theorem}
\begin{remark}
We note that we cannot take the limit $\nu\to 0$ in the a priori estimate above.
\end{remark}
\begin{proof}
Suppose $T_{max}<T$. Let $v=u-z$. Then $v$ satisfies system
\eqref{eqn:BurgersEquation-11}. Let us multiply $i$-th equation of
system \eqref{eqn:BurgersEquation-11} on
$v^i|v^i|^{p-2},i=1,\ldots,d$, take a sum w.r.t. $i$ and integrate
w.r.t. to time and space variable. We get for $t\geq t_0>0$
\begin{align}
|v(t)|_{\mathbb{L}^p(\Rd)}^p&+2\nu
\intl_{t_0}^t\intl_{\Rd}\suml_i|v|^{p-2}(s,x)|\nabla
v^i(s,x)|^2dxds\nonumber\\
&+(p-2)\nu
\intl_{t_0}^t\intl_{\Rd}\suml_i|v|^{p-4}(s,x)|(v^i(s,x),\nabla
v^i(s,x))|^2dxds\nonumber\\
&\leq |v(t)|_{\mathbb{L}^p(\Rd)}^p+\nu
p\intl_{t_0}^t\intl_{\Rd}\suml_i|v|^{p-2}(s,x)|\nabla
v^i(s,x)|^2dxds\nonumber\\
&\leq |v(t_0)|_{\mathbb{L}^p(\Rd)}^p+\intl_{t_0}^t\intl_{\Rd}|F(z)|^pdxds+(p-1)\intl_{t_0}^t\intl_{\Rd}|v|^pdxds\nonumber\\
&+p\intl_{t_0}^t|\nabla
z|_{\mathbb{L}^{\infty}}(s)\intl_{\Rd}|v|^pdxds+\intl_{t_0}^t\intl_{\Rd}|v|^p\diver(v+z)\,dxds.\label{eqn:EnergyIneq_aux_1}
\end{align}
The last term in the inequality \eqref{eqn:EnergyIneq_aux_1} can be estimated from above as follows
\begin{align}
\intl_{t_0}^t\intl_{\Rd}|v|^p\diver(v+z)\,dx\,ds&\leq \intl_{t_0}^t|\diver z|_{L^{\infty}(\Rd)}|v(s)|_{\mathbb{L}^p(\Rd)}^pds\nonumber\\
&+\intl_{t_0}^t|v|_{\mathbb{L}^{\infty}([t_0,t]\times\Rd)}\intl_{\Rd}|v|^{p-1}|\diver v| dxds\nonumber\\
&= \intl_{t_0}^t|\diver z|_{L^{\infty}(\Rd)}|v(s)|_{\mathbb{L}^p(\Rd)}^pds+|v|_{\mathbb{L}^{\infty}([t_0,t]\times\Rd)}
\intl_{t_0}^t\intl_{\Rd}|v|^{\frac{p}{2}}|v|^{\frac{p}{2}-1}|\diver v| dxds\nonumber\\
&\leq \intl_{t_0}^t|\diver z|_{L^{\infty}(\Rd)}|v(s)|_{\mathbb{L}^p(\Rd)}^pds+|v|_{\mathbb{L}^{\infty}([t_0,t]\times\Rd)}\nonumber\\
&\Big(\frac{1}{4\eps}\intl_{t_0}^t\intl_{\Rd}|v|^pdxds+\eps\intl_{t_0}^t\intl_{\Rd}|v|^{p-2}|\diver v|^2 dxds \Big),\,\,\eps>0,\label{eqn:KeyIneq_1}
\end{align}
where the last inequality follows from the Young inequality.
We can estimate $\mathbb{L}^{\infty}$ norm of $v$ by Feynman-Kac formula in the same way as in the torus case above i.e. we have
\begin{equation}
|v|_{L^{\infty}([\delta,t];\mathbb{L}^{\infty}(\Rd))}\leq
(|v(\delta)|_{\mathbb{H}^{1,p}(\Rd)}+|z|_{L^{\infty}(0,t;\mathbb{H}^{2,p}(\Rd))}^2)e^{|\nabla
z|_{L^1(0,t;\mathbb{L}^{\infty}(\Rd))}},\label{FeynmanKacEst_2}
\end{equation}
for any fixed $0<\delta\leq t<T_{max}$.
For $t\geq t_0>0$ denote
\[
Q(t_0,t)=(|v(t_0)|_{\mathbb{H}^{1,p}(\Rd)}+|z|_{L^{\infty}(0,t;\mathbb{H}^{2,p}(\Rd))}^2)e^{|\nabla
z|_{L^1(0,t;\mathbb{L}^{\infty}(\Rd))}}.
\]
Combining inequalities \eqref{eqn:EnergyIneq_aux_1}, \eqref{eqn:KeyIneq_1} and \eqref{FeynmanKacEst_2} we infer
that
\begin{align}
|v(t)|_{\mathbb{L}^p(\Rd)}^p&+\nu
p\intl_{t_0}^t\intl_{\Rd}|v|^{p-2}(s,x)|\nabla
v(s,x)|^2dxds\nonumber\\
&\leq |v(t_0)|_{\mathbb{L}^p(\Rd)}^p+\intl_{t_0}^t\intl_{\Rd}|F(z)|^pdxds+(p-1)\intl_{t_0}^t|v(s)|_{\mathbb{L}^p(\Rd)}^pds\nonumber\\
&+(p+1)\intl_{t_0}^t|\nabla
z(s)|_{L^{\infty}}|v(s)|_{\mathbb{L}^p(\Rd)}^pds+\eps Q(t_0,t)\intl_{t_0}^t\intl_{\Rd}|v|^{p-2}|\diver v|^2 dxds\nonumber\\
&+\frac{Q(t_0,t)}{4\eps}\intl_{t_0}^t|v(s)|_{\mathbb{L}^p(\Rd)}^pds,\,\,\eps>0
\end{align}
Put $\eps=\frac{\nu p}{2Q(t_0,t)}$. Then
\begin{align}
|v(t)|_{\mathbb{L}^p(\Rd)}^p&+\frac{\nu
p}{2}\intl_{t_0}^t\intl_{\Rd}|v|^{p-2}(s,x)|\nabla
v(s,x)|^2dxds\nonumber\\
&\leq |v(t_0)|_{\mathbb{L}^p(\Rd)}^p+\intl_{t_0}^t\intl_{\Rd}|F(z)|^pdxds\nonumber\\
&+\intl_{t_0}^t\Big((p+1)|\nabla
z(s)|_{\mathbb{L}^{\infty}}+(p-1)+\frac{Q(t_0,t)^2}{2\nu p}\Big)|v(s)|_{\mathbb{L}^p(\Rd)}^pds.
\end{align}
Now we apply the Gronwall Lemma to conclude that
\begin{align}
|v(t)|_{\mathbb{L}^p(\Rd)}^p&+\frac{\nu
p}{2}\intl_{t_0}^t\intl_{\Rd}|v|^{p-2}(s,x)|\nabla
v(s,x)|^2dxds\nonumber\\
&\leq \left(|v(t_0)|_{\mathbb{L}^p(\Rd)}^p+\intl_{t_0}^t\intl_{\Rd}|F(z)|^pdxds\right)\nonumber\\
&\cdot\exp\Big\{(p+1)|\nabla
z|_{L^1([t_0,t],\mathbb{L}^{\infty})}+(p-1)t+\frac{tQ(t_0,t)^2}{2\nu p}\Big\}.
\end{align}
Taking the limit $t\to T_{max}$ and $t_0\to 0$ we get contradiction.
\end{proof}
In the next Theorem we
will show that if a Beale-Kato-Majda type condition is satisfied i.e. vorticity is bounded then
then 
 a priori estimate holds uniformly in $\nu \geq 0$.
\begin{theorem}\label{thm:GlobalExistence_2}
Fix $p>d$ and $\theta\in (0,1)$. Assume that $u_0\in
\mathbb{L}^p(\Rd)$,
$f\in M^{2p}([0,T],\mathbb H^{3,2p}(\Rd))$,
$g\in M^{2p}([0,T],\gamma(H,\mathbb H^{4,2p}(\Rd)))$. Let $u\in
L^{\infty}([0,T_{max});\mathbb{L}^p(\Rd))$ be a strong maximal local solution of
Burgers equation. Assume also
that a.s.
\begin{equation}
\dela{\omega=}\curl u\in
L^{\infty}(0,T_{max};\mathbb{L}^{\infty}(\Rd)),\label{eqn:BoundedVorticityCond}
\end{equation}
and 
there exists $t_0\in (0,T)$ such that
\begin{equation}
 \diver u(t_0,\cdot) \in \mathbb{L}^{\infty}(\Rd).\label{eqn:BoundedConvergenceCond}
\end{equation}
Then $T_{max}=T$ and we have a.s.
\begin{eqnarray}
|u(t)|_{\mathbb{L}^p(\Rd)}^p&\leq&
C(|u_0|_{\mathbb{L}^p(\Rd)}^p+|F(z)|_{L^1(0,t;\mathbb{L}^p(\Rd))})\nonumber\\
&&\exp\left\{C(t-t_0)(  |\diver u(t_0,\cdot) |_{\mathbb{L}^{\infty}(\Rd)} +|\curl u|_{\mathbb{L}^{\infty}((0,t]\times\Rd)}\right.\nonumber\\
&&\left.
+|z|_{L^{\infty}(0,T;\mathbb{H}^{3,p}(\Rd))}^2(1+|u_0|_{L^p}+|z|_{\mathbb{H}^{2,p}(\Rd)}^2)e^{t|z|_{\mathbb{H}^{2,p}(\Rd)}})\right\}.\label{eqn:EnergyInequality}
\end{eqnarray}
\end{theorem}
\begin{proof}
The proof follows the lines of Theorem 2.2 in \cite{Goldys-Neklyudov} and is omitted.

\end{proof}
\begin{remark}
It is possible to construct random dynamical system corresponding to the solution of stochastic Burgers equation following the argument of the first name auhour and Yuhong Li  \cite{BrzYu}.
\end{remark}
\section{Gradient case}
In this section we will consider a particular case when the initial condition and force are potential. 
\begin{corollary}
Fix $p>d$. Assume that $\psi_0\in H^{1,p}(\mathcal{O})$ a.s.,
$U\in M^{2p}([0,T],H^{4,2p}(\mathcal{O}))$,
$V\in M^{2p}([0,T],\gamma(H, H^{5,2p}(\mathcal{O})))$. Then there
exists unique global solution $u\in
C(0,T;\mathbb{L}^p(\mathcal{O}))$ a.s. of equation
\[\left\{
\begin{aligned}
du &+ (u\nabla)udt=(\nu\triangle u+\nabla U)\,dt+\nabla Vdw_t\\
u(0) &= \nabla\psi_0.
\end{aligned}
\right.\]
Furthermore, if
$\psi_0,U,V$ are non random then for $\mathcal{O}=\mathbb{T}^d$ we have
\[
\mathbb{E}\supl_{s\in [0,t]}|u(s)|_{\mathbb{L}^p(\mathcal{O})}^p\leq C(|\psi_0|_{H^{1,p}},|U|_{L^1([0,t],H^{2,p})},|V|_{L^2([0,t],\gamma(H,H^{3,p}))}),
\]
and for $\mathcal{O}=\mathbb{R}^d$ we have
\[
\mathbb{E}\log(1+\supl_{s\in [0,t]}|u(s)|_{\mathbb{L}^p(\mathcal{O})}^p)\leq C(|\psi_0|_{H^{1,p}},|U|_{L^1([0,t],H^{2,p})},|V|_{L^2([0,t],\gamma(H,H^{3,p}))}),
\]
\end{corollary}
\begin{proof}
The first part follows immediately from Corollary \ref{thm:TorusGlobalExistence} and Theorem \ref{thm:GlobalExistence_2}. The second part follows from estimates \eqref{eqn:TEnergyInequality}, \eqref{eqn:EnergyInequality}
and Fernique Theorem. Indeed, if $U,V$ are non random then Ornstein-Uhlenbeck process $z$ has gaussian distribution in $L^2([0,T],\mathbb{H}^{1,p}(\mathcal{O}))\subset L^1([0,T],\mathbb{L}^{\infty}(\mathcal{O}))$.
\end{proof}
Consequently, since $u$ is a gradient of a certain function provided the initial condition and the force are gradients we can deduce the following corollary.
\begin{corollary}\label{cor:HamJacobiEquation}
Fix $p>d$ and $\nu>0$.  Assume that $\psi_0\in H^{1,p}(\mathcal{O})$ a.s.,
$U\in M^{2p}([0,T], H^{4,2p}(\mathcal{O}))$,
$V\in M^{2p}([0,T],\gamma(H, H^{5,2p}(\mathcal{O})))$. Then there
exists unique global solution $\psi^{\nu}\in
C(0,T;H^{1,p}(\mathcal{O}))$ a.s. of the equation
\begin{equation}\label{eqn:KDP}
\left\{
\begin{aligned}
d\psi^{\nu} &+ |\nabla\psi^{\nu}|^2dt=(\nu\triangle \psi^{\nu}+U)\,dt+VdW_t\\
\psi^{\nu}(0) &= \psi_0.
\end{aligned}
\right.
\end{equation}
Furthermore, if
$\psi_0,U,V$ are non random then for $\mathcal{O}=\mathbb{T}^d$ we have
\begin{equation}\label{eqn:AprioriEstim}
\mathbb{E}\supl_{s\in [0,t]}|\psi^{\nu}(s)|_{H^{1,p}(\mathcal{O})}^p\leq C(|\psi_0|_{H^{1,p}(\mathcal{O})},|U|_{L^1([0,t],H^{2,p}(\mathcal{O}))},|V|_{L^2([0,t],\gamma(H,H^{2,p}(\mathcal{O})))}),\;t\geq 0.
\end{equation}
and for $\mathcal{O}=\mathbb{R}^d$ we have
\begin{equation}\label{eqn:AprioriEstimRd}
\mathbb{E}\log(1+\supl_{s\in [0,t]}|\psi^{\nu}(s)|_{H^{1,p}(\mathcal{O})}^p)\leq C(|\psi_0|_{H^{1,p}(\mathcal{O})},|U|_{L^1([0,t],H^{2,p}(\mathcal{O}))},|V|_{L^2([0,t],\gamma(H,H^{2,p}(\mathcal{O})))}),\;t\geq 0.
\end{equation}
\end{corollary}
We can notice that the estimates \eqref{eqn:AprioriEstim}, \eqref{eqn:AprioriEstimRd} are uniform
w.r.t. $\nu$. This leads us to the following Corollary.
\begin{corollary}\label{cor:HamJacobiEquation-2}
Fix $p>d$. Assume that $\psi_0\in H^{1,p}(\mathcal{O})$ a.s.,
$U\in L^1([0,T],H^{2,p}(\mathcal{O}))$,
$V\in L^p(0,T;\gamma(H,H^{2,p}(\mathcal{O})))$. Then there
exists unique global viscosity solution $\psi\in
C(0,T;H^{1,p}(\mathcal{O}))$ of the equation
\begin{equation}\label{eqn:KDP-2}
\left\{
\begin{aligned}
d\psi &+ |\nabla\psi|^2dt=Udt+VdW_t\\
\psi(0) &= \psi_0.
\end{aligned}
\right.
\end{equation}
and for $\mathcal{O}=\mathbb{T}^d$ we have
\begin{equation}\label{eqn:AprioriEstim-2}
\mathbb{E}\supl_{s\in [0,t]}|\psi(s)|_{H^{1,p}(\mathcal{O})}^p \leq C(|\psi_0|_{H^{1,p}},|U|_{L^1([0,t],H^{2,p})},|V|_{L^2([0,t],\gamma(H,H^{2,p}))}),\;t\geq 0.
\end{equation}
Furthermore, for $\mathcal{O}=\mathbb{R}^d$ we have
\begin{equation}\label{eqn:AprioriEstim-2RD}
\mathbb{E}\log(1+\supl_{s\in [0,t]}|\psi(s)|_{H^{1,p}(\mathcal{O})}^p) \leq C(|\psi_0|_{H^{1,p}},|U|_{L^1([0,t],H^{2,p})},|V|_{L^2([0,t],\gamma(H,H^{2,p}))}),\;t\geq 0.
\end{equation}
\end{corollary}
\begin{remark}
The Corollaries \ref{cor:HamJacobiEquation} and \ref{cor:HamJacobiEquation-2} are different from results of \cite{DirrSouganidis-2005} because they consider viscosity solutions in the space of continuous functions while we consider solutions in $\mathbb{H}^{1,p}(\mathcal{O})$, $p>d$.
\end{remark}
\begin{proof}
Let $\{\psi^{\nu}\}_{\nu>0}\in C(0,T;\mathbb{H}^{1,p}(\mathcal{O}))\cap
C^{1,2}((0,T]\times \mathcal{O})$ be sequence of solutions of the
equation \eqref{eqn:KDP}. Since
$\mathbb{H}^{1,p}(\mathcal{O})\subset C(\mathcal{O},\mathbb R^d),p>d$ and estimate
\eqref{eqn:TEnergyInequality} (corr. estimate \eqref{eqn:EnergyInequality} if $\mathcal{O}=\Rd$) we have uniform w.r.t. $\nu$ estimate $\mathbb{P}$-a.s.
\begin{eqnarray}
|\psi^{\nu}|_{C(0,T;C(\mathcal{O},\mathbb R^d))}^p\leq
K(T,\psi_0,h,d),T>0,p>d.\label{eqn:AprioriEstim-3}
\end{eqnarray}
Then according to Theorem 1.1, p. 175 in \cite{Barles-2006} we
have that there exist uniformly bounded upper continuous
subsolution $\psi^*=\limsupl_{\nu\to 0}^*\psi^{\nu}$ $\mathbb{P}$-a.s. and uniformly
bounded lower continuous supersolution $\psi_*=\liminfl_{\nu\to
0}^*\psi^{\nu}$ $\mathbb{P}$-a.s. of equation \eqref{eqn:KDP-2}.
Therefore, by comparison principle for viscosity solutions of
Hamilton-Jacobi equations (see Theorem 2, p. 585 and Remark 3, p. 593 of \cite{CrIshLi-1987}), $\psi^*\leq\psi_*$ and
$\psi=\psi^*=\psi_*$. Thus, $\psi^{\nu}$ locally uniformly
converges to unique viscosity solution $\psi$ of equation
\eqref{eqn:KDP-2} $\mathbb{P}$-a.s. Estimate
\eqref{eqn:AprioriEstim} implies that $\psi$ satisfies
\eqref{eqn:AprioriEstim-2}.

\end{proof}

\end{document}